\documentclass[10pt,oneside]{amsart}

\usepackage{tikz}
\usetikzlibrary{automata, positioning, arrows}

      \usepackage{amssymb}
\usepackage{amsfonts,amstext}
\usepackage{graphicx}
\usepackage{url}
\usepackage{cancel}
\usepackage{comment}

\usepackage{amsmath,amstext,amssymb,amscd}
\usepackage{verbatim} 
\usepackage{mathtools}
      \usepackage{tabularx,booktabs}

\usepackage{amssymb,latexsym}
\usepackage{amsmath}
\usepackage{amsthm}
\usepackage{array}
\usepackage{amssymb}
\usepackage{multicol}
\usepackage{makecell}
\newtheorem{thm}{Theorem}[section]

\newtheorem{prop}[thm]{Proposition}

\theoremstyle{remark}

\theoremstyle{definition}

\newtheorem{ex}[thm]{Example}

\usepackage{amsthm}
 \usepackage{multicol}
\allowdisplaybreaks
\theoremstyle{definition}

\newtheorem{definition}{Definition}

\usepackage{multicol}

      \newcommand{\ra}{\rightarrow}

      \makeatletter
      \def\@setcopyright{}
      \def\serieslogo@{}
     
      \makeatother

\begin{document}

\author{David McCune}
\address{David McCune,  Department of Mathematics and Data Science, William Jewell College, 500 College Hill, Liberty, MO, 64068-1896}
\email{mccuned@william.jewell.edu} 

\author{Erin Martin}
\address{Erin Martin, Department of Mathematics, Brigham Young University, Provo, UT 84602}
\email{martin@mathematics.byu.edu} 

\author{Grant Latina}
\address{Grant Latina, William Jewell College, 500 College Hill, Liberty, MO, 64068-1896}
\email{latinag.19@william.jewell.edu}

\author{Kaitlyn Simms}
\address{Kaitlyn Simms, William Jewell College, 500 College Hill, Liberty, MO, 64068-1896}
\email{simmsk.19@william.jewell.edu}

\title[ A Comparison of Sequential RCV and STV ]{A Comparison of Sequential Ranked-Choice Voting and  Single Transferable Vote}

\begin{abstract}
The methods of single transferable vote (STV) and sequential ranked-choice voting (RCV) are different methods for electing a set of winners in multiwinner elections. STV is a classical voting method that has been widely used internationally for many years. By contrast, sequential RCV has rarely been used, and only recently has seen an increase in usage as several cities in Utah have adopted the method to elect city council members. We use Monte Carlo simulations and a large database of real-world ranked-choice elections to investigate the behavior of sequential RCV by comparing it to STV. Our general finding is that sequential RCV often produces different winner sets than STV. Furthermore, sequential RCV is best understood as an excellence-based method which will not produce proportional results, often at the expense of minority interests.\end{abstract}



\maketitle

\section{Introduction}
Over the last 20 years, various forms of ranked-choice voting (RCV) have been adopted for political elections across the United States. The single-winner version of RCV (referred to in the social choice literature as instant runoff voting, the plurality elimination rule, the Hare method, etc.) has been used for municipal elections in cities such as Minneapolis, MN, Oakland, CA, and San Francisco, CA. Single-winner RCV has also been used for statewide and federal elections in Alaska and Maine. A multiwinner version of RCV known as single-transferable vote (STV) has been used for some political elections in the US, though not nearly as widely. A version of STV has been used for multiwinner city council and school committee elections in Cambridge, MA since 1941; more recently, STV has been used for some municipal elections in Albany, CA, Eastpointe, MI, and Minneapolis, MN. For a survey of the different ranked-choice voting methods currently used in the US, see \cite{San}.

In 2019, the cities of Payson and Vineyard in Utah introduced a new multiwinner voting method, called sequential RCV, into the US political landscape. Following the trend of other cities which have adopted forms of RCV for municipal elections, the residents of Payson and Vineyard decided to use sequential RCV to fill seats on their city councils. To the best of our knowledge, prior to its use in Utah sequential RCV has  been used only for a handful of political elections in Australia, most of which occurred prior to 1948 \cite{F, S}. In 2021, the use of sequential RCV expanded to several more cities in Utah. A new method calls for rigorous analysis, which in this case has not been  done previously because of the infrequency of the method's use.

In this paper we analyze sequential RCV primarily by comparing it to STV, a method designed to achieve proportional representation and the most common real-world multiwinner voting method which uses preference ballots. Our general finding is that sequential RCV does not produce results that agree with proportional representation, instead choosing winner sets which align with an excellence-based paradigm. In particular, sequential RCV tends to select the Condorcet committee (as defined in \cite{G} and \cite{R}) as the winner set, thereby giving winner sets in which the candidates are similar to each other (rather than providing a broad base of representation to the electorate). This aspect of sequential RCV has been anticipated in the popular press \cite{H} and in the book \cite{F}; our contribution is that we provide mathematical and empirical analysis. We use elections generated by Monte Carlo simulation with a large number of runs, as well as a sizable database of recent multiwinner elections, to support such claims.


There is a large social choice literature which compares the performance of different single-winner voting methods. Historically, most of this literature was theoretical (see \cite{CGZ, G, GL, WP} for recent examples), mainly due to the lack of vote data and the computing power to process it. In recent years the social choice empirical literature on single-winner voting methods has increased \cite{DGK, GSZ, MMa, PPR, RKH, Song} as more ranked-choice vote data has become available. The social choice literature regarding multiwinner voting methods is much smaller, presumably because such methods are harder to study from a theoretical point of view. There has recently been an increased interest in multiwinner voting methods and their properties \cite{CJMW, EFSS, FSST, S-FF}, although  studies which specifically address the kinds of questions analyzed in this article are rare. An example of a study similar to ours which compares the frequency with which two multiwinner methods agree is \cite{DKT}.  Their work primarily uses Monte Carlo simulation to compare the performance of a Chamberlin-Courant voting method to other classical multiwinner methods such as $k$-plurality and $k$-Borda. We also use simulations to obtain theoretical results but we supplement our work  with empirical results using a large database of real-world multiwinner political elections from Scotland. \cite{EFLSST} studies how different multiwinner voting methods produce different winner sets over a two-dimensional Euclidean domain.

The paper is structured as follows. In Section \ref{prelim} we provide our notation and define STV and sequential RCV. Section \ref{multiwinner_wants} contains a brief discussion of the potential goals of multiwinner voting methods. In Section \ref{simulations} we provide theoretical results comparing sequential RCV and STV using Monte Carlo simulations. Sections \ref{scottish_data} and \ref{utah_data} present empirical results using multiwinner ballot data from a large database of Scottish STV elections and a small set of sequential RCV elections from Utah. Sections \ref{simulations}, \ref{scottish_data}, and \ref{utah_data} compare the behavior of the two methods using concepts such as a Condorcet committee (defined below).  Section \ref{conclusion} concludes.

\section{Preliminaries: Sequential RCV and STV}\label{prelim}

Let $n$ represent the number of candidates in an election, $S$ the number of seats available (which is the size of the winner set), and $V$ the number of voters. In the elections we study voters cast their votes using preference ballots, which allow each voter to provide a preference ranking of the candidates. After the election the ballots are aggregated into a \emph{preference profile}, which shows the number of each type of preference ballot cast. Table \ref{profile} provides an example of a preference profile with four candidates $A$, $B$, $C$, and $D$. The number 1801 denotes that 1801 voters cast a ballot ranking $A$ first, $B$ second, $C$ third, and $D$ fourth. The other numbers across the top row convey similar information about the number of ballots cast with the given ranking. We use capital letters to denote candidates but reserve the letter $P$ for a preference profile. We refer to a pair $(P, S)$ as an \emph{election}.

In a theoretical electoral setting such as what we analyze in Section \ref{simulations}, we assume each voter provides a complete ranking of the candidates. In practice voters often provide only a partial ranking, sometimes ranking only a single candidate on the ballot. The preference profile in Table \ref{profile} contains many  partial ballots.

\begin{table}[htb]
  \centering

\begin{tabular}{l|c|c|c|c|c|c|c|c|c}
Num. Voters & 1799 & 1801 & 100 & 901 & 900 & 498 & 2000 & 1400 & 601\\
\hline
1st choice & $A$ & $A$ & $A$ & $B$ & $B$& $C$& $C$ & $D$ &$D$\\
2nd choice & $B$ & $B$ & $C$ & $C$ & $D$ & $B$ & $D$ & $B$ & $C$\\
3rd choice & & $C$& $D$ &$A$ & & $D$ & $A$ & &\\
4th choice & &$D$ &  & $D$&  &$A$&&&\\
\end{tabular}

\caption{An example of a preference profile with $n=4$ and $V= 10000$.}
  \label{profile}
  \end{table}

Before defining STV and sequential RCV, we first define the method of instant runoff voting (IRV), which is the single-winner version of both multiwinner methods. That is, when $S=1$ both STV and sequential RCV output the IRV winner. We note that in popular discourse the term ``ranked-choice voting'' is often used interchangeably with IRV; we choose the term that is more common in the social choice literature.

\begin{definition}
 Under the method of \textbf{instant runoff voting}, if a candidate receives an initial majority of first-place votes then that candidate is the winner. Otherwise the candidate with the fewest first-place votes is eliminated and their votes are transferred to the next candidate on each of their ballots who has not been previously eliminated. This process continues until a surviving candidate achieves a majority of the remaining first-place votes, at which point the process stops and this candidate is declared the winner.
\end{definition}

We illustrate this definition using the preference profile from Table \ref{profile}.

\begin{ex}\label{IRV_example}
In Table \ref{profile} the initial first-place vote totals for $A$, $B$, $C$, and $D$ are 3700, 1801, 2498, and 2001, respectively. No candidate achieves a majority and thus $B$ is eliminated.  As a result 901 votes are transferred to $C$ and 900 are transferred to $D$, creating updated first-place votes totals of 3700, 3399, and 2901 for $A$, $C$, and $D$ respectively. No candidate achieves a majority and thus $D$ is eliminated. As a result 601 votes are transferred to $C$ and the 2300 votes corresponding to the voters who ranked only $B$ and $D$ on their ballots are removed from the election. (These ballots are said to be \emph{exhausted} and cannot be transferred to either $A$ or $C$.) Candidate $C$ now defeats $A$ with 4000 votes to 3700, and $C$ is the IRV winner.
\end{ex}

In practice, the preference profile of an election is usually too large to display and thus election offices display how the election unfolds using a \emph{votes-by-round table}. The left table of Table \ref{votes_IRV} shows how the instant runoff election for the preference profile in Table \ref{profile} unfolds.

We now define the methods of sequential RCV and STV.

\begin{definition}
 Given an election $(P, S)$, \textbf{sequential RCV} chooses a winner set of size $S$ as follows. The first seat is given to the IRV winner; call this candidate $W_1$. Remove $W_1$ from all ballots in the original preference data in $P$  and give the second seat to the IRV winner under the modified data; call this candidate $W_2$. Now remove both $W_1$ and $W_2$ from all ballots in the original preference data in $P$ and give the third seat to the IRV winner under the modified vote data with both previous winners removed. Continue this process until all $S$ seats are filled.
\end{definition}

Sequential RCV is perhaps better named ``iterated IRV'' but we use the term that is used by the Utah election offices which implemented the method.

\begin{ex}
To calculate the sequential RCV winner set for the election $(P, 2)$, where $P$ is the preference profile from Table \ref{profile}, we first find the IRV winner. By Example \ref{IRV_example},  the first seat goes to $C$. To calculate the second winner we remove $C$ from all ballots in Table \ref{profile}, shifting candidates up one slot to fill the empty slots left by $C$, and then run IRV again. The resulting votes-by-round table is shown in the right table of Table \ref{votes_IRV}. With candidate $C$ removed the initial vote totals for $B$ and $D$ both increase, but $B$ still has the fewest first-place votes and is eliminated first. As a result, 901 votes are transferred to $A$ (note that $A$ receives these votes because $C$ is removed) and 900 are transferred to $D$, securing a majority for $D$. Thus the sequential RCV winner set for $(P, 2)$ is $\{C,D\}$. 

If we wished to elect three winners then we would remove both $C$ and $D$ from all ballots in Table \ref{profile} and again use IRV. In this case there are only two candidates remaining and $A$ would win the election in one round because $A$ starts the modified election with an extra 2000 first-place votes from the voters who ranked $C$ first, $D$ second, and $A$ third. Thus the winner set under sequential RCV for the election $(P, 3)$ is $\{A, C, D\}$.
\end{ex}

\begin{table}[htb]
  \centering
  
  \begin{tabular}{ccc}
  
  \begin{tabular}{c|ccc}
  \hline
  Candidate & \multicolumn{3}{c}{Votes by Round}\\
  \hline
  $A$ & 3700 & 3700 & 3700 \\
  $B$ & 1801 &          &           \\
  $C$ & 2498 & 3399 & 4000\\
  $D$ & 2001 & 2901 & \\
  \hline
  \end{tabular}
  
  & &
  
    \begin{tabular}{c|cc}
  \hline
  Candidate & \multicolumn{2}{c}{Votes by Round}\\
  \hline
  $A$ & 3700 & 4601  \\
  $B$ & 2299 &                     \\
  $\xcancel{C}$ & &\\
  $D$ & 4001 & 5399 \\
  \hline
  \end{tabular}\\
  
  \end{tabular}

\caption{(Left) The votes-by-round table for the preference profile in Table \ref{profile} when using instant runoff voting. (Right) The resulting votes-by-round table when removing $C$ from the preference profile and using instant runoff voting. }
  \label{votes_IRV}
  \end{table}

 To the best of our knowledge sequential RCV has rarely been used for political elections and therefore has not been well-studied. The method was used for some elections in Australia under the name ``preferential block voting'' \cite{S}, but otherwise we can find no record of the method's use outside the US. Election officials in Utah seem to have invented sequential RCV and, since such officials are not in the habit of explaining their thinking, the method's motivation is not clear. Presumably the method's creators were looking for a quick way to generalize IRV to the multiwinner setting. Other than in Australia and Utah we are aware of only one other use of this method, when it was used for a municipal election in Portland, Maine in 2021. In this case the method was referred to as ``multi-pass instant runoff voting'' in city communications. The city stumbled into the method's use as the city charter insisted that candidates who win a ranked-choice election must secure a majority of the vote, and thus sequential RCV seemed like the required voting method for a multiwinner election in 2021.
  
  By contrast, the method of single transferable vote (STV) is widely used across the world, and is by far the most popular voting method for multiwinner political elections which use preference ballots. For example, STV is used for local government elections in Ireland, Northern Ireland, and Scotland, as well as elections for the national legislature in Australia, Ireland, and Malta. It is difficult to provide a complete definition of STV in a concise fashion. Therefore, we provide a high level description which we illustrate using  the preference profile in Table \ref{profile}. The formal description of the rules as implemented in Scotland (the source the empirical data in this paper) can be found at \url{https://www.legislation.gov.uk/sdsi/2007/0110714245}.

The method of STV proceeds in rounds. In each round, either a candidate earns enough votes to be elected or no candidate is elected and the candidate with the fewest (first-place) votes is eliminated. The number of votes required to be elected is called the \emph{quota}, and is calculated by \[\text{quota } = \left\lfloor \frac{V}{S+1}\right\rfloor +1.\]

If no candidate reaches quota in a given round then the candidate with the fewest first-place votes is eliminated, and this candidate's votes are transferred to the next candidate on their ballots who has not been elected or eliminated. If a candidate reaches quota, that candidate is elected and the votes they receive above quota (\emph{surplus votes}) are transferred in a fashion similar to that of an eliminated candidate, except the surplus votes are transferred in proportion to the number of ballots on which each other candidate appears.  To explain how these transfers work, suppose candidate $A$ is elected with a total of $a$ votes and a surplus of $A_s$ votes (so that $A_s = a - $ quota), and candidate $B$ is the next eligible candidate on $b$ of these ballots. Rather than receive $b$ votes from the election of $A$ candidate $B$ receives $(A_s/a)b$ votes, resulting in a fractional vote transfer. The method continues in this fashion until $S$ candidates are elected, or until some number $S'<S$ of candidates have been elected by surpassing quota and there are only $S-S'$ candidates remaining who have not been elected or eliminated.

\begin{ex}\label{STV_example}
Let $P$ be the preference profile from table \ref{profile}. We find the winner set of the election $(P, 2)$ under STV.

Since $S=2$, the quota is $\lfloor 10000/3 \rfloor + 1 = 3334$. Candidate $A$ achieves quota in the first round and their $3700-3334 = 366$ surplus votes must be transferred. Candidate $B$ is ranked second on 3600 of the ballots on which $A$ is ranked first and therefore receives $(3600/3700)366=356.11$ of $A$'s surplus votes. The remaining $(100/3700)366=9.89$ surplus votes are transferred to $C$. At this point in the vote-transfer process (see Table \ref{votes_STV}) none of the remaining candidates achieves quota. Candidate $D$ is eliminated and the 1400 votes transferred to $B$ secures the second seat for $B$. Thus the winner set under STV is $\{A,B\}$, disjoint from the winner set $\{C,D\}$ under sequential RCV.

If an election contains more candidates or if more seats are available, the method continues in a similar fashion until all seats are filled.
\end{ex}
  
  \begin{table}[htb]
  \centering
  
   \begin{tabular}{c|ccc}
  \hline
  Candidate & \multicolumn{3}{c}{Votes by Round}\\
  \hline
  $A$ & \textbf{3700} &  &  \\
  $B$ & 1801 &    2157.11      &      \textbf{3557.11}     \\
  $C$ & 2498 & 2507.89 & 3108.89\\
  $D$ & 2001 & 2001 & \\
  \hline
  \end{tabular}
  
  \caption{The votes-by-round table for the preference profile in Table \ref{profile} when using STV with 2 seats. The bold numbers indicate when a candidate achieves quota.}
  \label{votes_STV}
  \end{table}

As the above definitions and examples suggest, these two voting methods behave quite differently and can produce winner sets with little or no overlap. We close this section by showing that if the number of candidates is large relative to the number of seats, the two methods can produce disjoint winner sets. Thus the choice between these two methods can have profound electoral effects. We prove the proposition for the case $n=2S$ but note the result holds for $n\ge 2S$ because we can always add additional candidates who receive negligible voter support. 

\begin{prop}\label{first_prop}
Let $S\ge 2$ and $n=2S$. Then there exists an election $(P, S)$ such that the set of winners under STV is disjoint from the set of winners under sequential RCV.
\end{prop}

\begin{proof}
We first construct an election with $S+1$ candidates in which the instant runoff winner is not a member of the winner set under STV with $S$ seats. It is then straightforward to construct an election as described in the Proposition statement for $S=2n$.

Label the candidates $C_1, \dots, C_{S+1}$ where the candidates are ordered so that $C_1$ has the most first-place votes, $C_2$ has the second-most first-place votes, and so on. Let $C_1$ have a vote share of $0.5-\epsilon$ for small $\epsilon$ and distribute the rest of the first-place votes to the other candidates so that their first-place vote shares are roughly equal, which ensures that $C_S$ does not surpass quota in the first round. Distribute $C_1$'s second-place votes to $C_2, \dots, C_{S-1}, C_{S+1}$ so that each of these candidates achieves quota upon the election of $C_1$. Note this is possible because $C_S$'s initial first-place vote total is smaller than quota. For each voter who ranks the candidates $C_2, \dots, C_{S-1}, C_{S+1}$ first, give $C_S$ the voter's second ranking. The rest of the rankings do not matter and can be completed in an arbitrary manner.

For such a construction, we claim that the instant runoff winner is $C_S$, who is not a member of the STV winner set with $S$ seats. To see this, note that under the instant runoff algorithm $C_{S+1}$ is eliminated first, and all of their votes are transferred to $C_S$. Since the bottom $S$ candidates all have roughly the same initial first-place vote share, $C_S$ now has more first-place votes than $C_{S-1}$, who is eliminated and all of their votes are transferred to $C_S$. This process continues until $C_S$ defeats $C_1$ head-to-head in the final round because $C_1$'s vote share of $0.5-\epsilon$ never surpasses 0.5 under the instant runoff algorithm by construction. Furthermore, if we run the STV algorithm with $S$ seats then $C_1$ achieves quota in the first round and their surplus votes are distributed so that all candidates other than $C_S$ achieve quota in the second round. Thus $C_S$ does not win a seat under STV.

Table \ref{prop_table} illustrates one possible realization of this construction for $S=4$ and 10,000 voters (note that quota in this election is 2001 votes). As $S\rightarrow\infty$ the construction will eventually require more than 10,000 voters, but this is not an issue if we let the size of the electorate also approach infinity. Note that in this election the STV winner set is $\{C_1, C_2, C_3, C_5\}$ because when $C_1$ achieves quota in the first round, the surplus 2997 votes are split equally among $C_2$, $C_3$, and $C_5$, causing each of these candidates to surpass quota in the second round. However, the instant runoff winner is $C_4$ because under the instant runoff algorithm $C_5$ is eliminated first, transferring 1249 votes to $C_4$. $C_3$ is eliminated next, followed by $C_2$, and in the final round $C_4$ defeats $C_1$ with 5002 votes to 4998.

 \begin{table}[htb]
  \centering

\begin{tabular}{l|c|c|c|c|c|c|c}
Num. Voters & 1666 & 1666 & 1666 & 1252 & 1251 & 1250 & 1249\\
\hline
1st choice & $C_1$ & $C_1$ & $C_1$& $C_2$& $C_3$ & $C_4$ & $C_5$\\
2nd choice & $C_2$ & $C_3$ & $C_5$ & $C_4$ & $C_4$ & $X$ & $C_4$\\
 $\vdots$ &&&&&&\\
\end{tabular}

 \caption{A preference profile in which the IRV winner is not a member of the STV winner set of the election $(P,4)$.}
  \label{prop_table}
  \end{table}

To create an election as described in the Proposition statement, we add candidates $C_{S+2}, \dots, C_n$ to the election so that on every ballot $C_{S+2}$ is ranked just below $C_S$ and $C_{S+i}$ is ranked just below $C_{S+i-1}$ for all $3\le i \le S$. Then the winner set under STV is still $\{C_1, \dots, C_{S-1}, C_{S+1}\}$ and the winner set under sequential RCV is $\{C_S, C_{S+2}, \dots, C_{2S}\}$. The latter statement is true because after $C_S$ wins the first seat under IRV they are replaced on all ballots by $C_{S+2}$, who then wins the next seat under IRV, etc. \end{proof}

\section{What Do We Want from a Multiwinner Voting Method?}\label{multiwinner_wants}

The goal of a single-winner voting method is generally quite clear: we wish to select the single ``most deserving'' candidate. Of course, different methods disagree about what is meant by \emph{most deserving}; for example, some methods (such as IRV) insist that a candidate who receives a majority of first-place votes must be the winner while other methods (such as Borda count) do not make such a requirement. Even though different single-winner methods disagree on the evaluation of the single best candidate, the overall goal of these methods is the same. As nicely outlined in \cite{FSST}, this is no longer the case for multiwinner voting methods. When choosing a set of winners of size at least two, the objective of a method can change based on context. For example, if  sports coaches were to hold an election to choose the top three athletes in their league so that these athletes receive a prize, the voting method used should try to select the ``most deserving'' three candidates. In the language of \cite{FSST}, this context calls for an excellence-based voting method. However, in a political setting in which proportional representation is the goal we no longer necessarily care about the most excellent candidates. To see why, suppose that the three ``most deserving'' candidates (by some measure) are extremely similar to each other in their political ideology. If we were to elect a winner set of size three, it would generally be undesirable to elect these three similar candidates if our goal is proportional representation.

The tension between excellence and proportionality can perhaps best be illustrated by vote data as shown in Table  \ref{small_profile}. If $P$ is the profile in the table then the winner set for the election $(P, 2)$ under STV is $\{A,C\}$. This outcome achieves proportional representation since the first-place votes are split roughly equally between $A$ and $C$. However, note that more voters prefer $B$ to $C$ than the reverse, and thus in an excellence-based setting (such as giving an award to two athletes) a strong argument can be made that the appropriate winner set is $\{A,B\}$, which is the winner set of $(P, 2)$ under sequential RCV.

\begin{table}[htb]
  \centering
  
   \begin{tabular}{l|cc}
   Num. Voters & 52 & 48\\
   \hline
   1st choice & $A$ & $C$\\
   2nd choice & $B$ & $D$\\
   3rd choice & $C$ & $A$\\
   4th choice & $D$ & $B$\\
   \end{tabular}
    \caption{A preference profile which illustrates the fundamental difference between sequential RCV and STV.}
  \label{small_profile}
  \end{table}

There are  other multiwinner electoral settings than those that call for an excellence-based voting method or a method that attempts to achieve proportional representation, but these two settings are the focus of this article. STV attempts to achieve proportional representation by construction \cite{D}; sequential RCV, on the other hand, seems to be an excellence-based method. To measure this tension, we use the notion of a \emph{Condorcet committee} as defined in \cite{G} and \cite{R}; such a set of candidates is a natural generalization of the classical single-winner notion of a Condorcet winner, a candidate who defeats all other candidates in a head-to-head matchup.

\begin{definition}
A subset of candidates $\mathcal{C}$ is a \textbf{Condorcet committee} if for each pair of candidates $(A,B)$ where $A \in \mathcal{C}$ and $B \not\in\mathcal{C}$, more voters prefer $A$ to $B$ than prefer $B$ to $A$.
\end{definition}

The original definition of a Condorcet committee assumes that all voters provide a complete ranking of the candidates. We do not assume such a setting, and thus we must choose how to handle head-to-head comparisons when one or both of the candidates under consideration are left off some ballots. We choose to interpret partial ballots under the \emph{weak order model} \cite{PPR} wherein all candidates left off the ballot are assumed to be tied for last place on these ballots. This is the model of partial ballots which is used by election offices that implement ranked-choice elections, although such offices do not use this language.

In Table \ref{small_profile}, the Condorcet committee is $\{A,B\}$ because each of these candidates wins a head-to-head matchup against $C$ or $D$. Note that a Condorcet committee of size $S$ does not necessarily exist; the reader can check that for the vote data in Table \ref{profile} a Condorcet committee of size two does not exist. Because sequential RCV is an excellence-based method and a Condorcet committee is an excellence-based notion of a committee of a certain size, we hypothesize that when such a committee exists it is more likely to be selected by sequential RCV than STV.  In subsequent sections we use the notion of a Condorcet committee to measure the degree to which each voting method produces an excellence-based result.

In addition to the concept of a Condorcet committee, we consider two other ways of comparing the performance of the two voting methods. These notions measure the degree to which voters achieve some level of representation in the winner set; we assume that a voter's $S$ top-ranked candidates on their ballot is the voter's desired winner set. To motivate the definitions, note that if we choose a winner set of $\{A,B\}$ for the election in Table \ref{small_profile} then 48\% of the electorate receives none of their top two candidates in the winner set, and in some sense these voters are unrepresented. On the other hand, if the winner set were $\{A,C\}$ or even $\{A,D\}$ then 100\% of the electorate is represented by at least one of the candidates they wanted in the winner set, and therefore a winner set of $\{A,C\}$ achieves proportional representation better than $\{A,B\}$. On the other hand, it the winner set is $\{A,B\}$ then 52\% of the electorate achieves their desired winner set, getting everything they want, while 48\% of the electorate receives nothing. A winner set of $\{A,B\}$ maximizes the percentage of voters who receive precisely the winner set they desire; in a proportional representation setting, we do not wish to maximize this percentage. These notions of minimizing the percentage of voters who receive no representation or maximizing the percentage of voters who receive maximal representation can help elucidate the differences between sequential RCV and STV.

\begin{definition}
Given an election $(P, S)$ and a winner set $\mathcal{W}$, the election's \textbf{degree of misrepresentation} is the percentage of voters for whom none of the top $S$ candidates on their ballots are in $\mathcal{W}$.
\end{definition}

While we are unaware of others making a definition equivalent to this, the degree of misrepresentation is inspired by the Chamberlin-Courant voting rule when using an approval-based weighting vector consisting of $S$ ones follows by $n-S$ zeroes \cite{CC}. See \cite{DKT} for an approachable introduction to the Chamberlin-Courant method. Intuitively, a method based on proportionality should try to choose a winner set which minimizes the degree of misrepresentation, ensuring that as many voters as possible are represented by at least one candidate whom they wanted in the winner set. In an excellence-based context such an outcome may not be important.

\begin{definition}
Given an election $(P, S)$ and a winner set $\mathcal{W}$, the election's \textbf{degree of maximal representation} is the percentage of voters for whom all of the top $S$ candidates on their ballots are in $\mathcal{W}$.
\end{definition}

The degree of maximal representation has not been previously defined, presumably because no voting method is built with the purpose of maximizing this percentage (although this may be an unintended consequence of using excellence-based methods). Because sequential RCV is an excellence-based method while STV tries to achieve proportionality, we hypothesize that sequential RCV will generally produce higher degrees of either type than STV.

If a voter casts a ballot with fewer than $S$ candidates ranked then we must decide whether to count this voter toward the degree of maximal representation if every candidate ranked on the ballot wins a seat. The spirit of the maximal representation definition is that we want to measure the amount of voters who get everything they want, and therefore we choose to count such a voter under our definition.

We now present our results, starting with theoretical results obtained from Monte Carlo simulations.

\section{Simulation Results}\label{simulations}

In this section we estimate the probabilities that the methods of sequential RCV and STV produce the same set of winners for various choices of $n$ and $S$ under two random models of voter behavior. We also estimate the conditional probability that each method selects the Condorcet committee, assuming one exists, and we investigate the degrees of misrepresentation and maximal representation under the models. Our estimates are obtained from Monte Carlo simulations where we use the random models to create the ballot data for a given election. The two models are the \emph{impartial culture model} and the \emph{impartial anonymous culture model}, both of which are classical random models which are widely used in the theoretical social choice literature. In the impartial culture model each voter chooses one of the $n!$ preferences rankings at random, essentially voting by rolling a die. In the impartial anonymous culture model, the vote distribution across the top of the preference profile is chosen at random.

\begin{definition}
Under the \textbf{impartial culture (IC) model}, each of the $V$ voters chooses one of the $n!$ complete candidate rankings at random.
\end{definition}

\begin{definition}
Under the \textbf{impartial anonymous culture (IAC) model} the list of numbers $(v_1, \dots, v_{n!})$ across the top of the preference profile is chosen at random, where the candidate rankings in each column are ordered lexicographically. That is, the preference profile is determined by choosing a vote distribution $(v_1, \dots, v_{n!})$ at random under the constraints $v_i\ge 0$, $\sum v_i = V$.
\end{definition}

  Generating a preference profile at random under the IC model is conceptually straightforward, as we simply let each of the $V$ voters roll an $n!$-sided die. Generating an election at random under the IAC model for a fixed $V$ is more challenging. For this task we use code developed by Awde et al. \cite{ADKRT}, which uses an urn-based probabilistic model. We note that if we wish to choose a percentage distribution of votes at random, letting $V\ra \infty$, then the task of generating an election at random under the IAC model is much simpler (see \cite{LLV}, for example).

Our simulations work as follows. For each choice of $n \in \{3,4,5,6\}$ and $S \in \{2, \dots, n-1\}$, for each model we generate 100,000 elections at random using $V=1001$. For each generated election we check if the sequential RCV winner set equals the STV winner set, and we check if each winner set equals the Condorcet committee (if one exists).  Due to limits of computation time we did not run simulations for $n\ge 7$. All code is written in Python 3.

The estimated probabilities given by the simulations are shown in Table \ref{IC_IAC_simulations}. For example, the table shows that if $n=3$ and $S=2$ then the two methods chose the same winner set in 78.8\% of the generated elections under the IC model. The column ``1 Winner Diff.'' shows the percentage of elections in which neither method produced a tie and the winner sets differed by only a single candidate. For elections in which $S<n-1$ it is possible that the winner sets differ by more than one candidate; the ``2 Winners Diff.'' column shows the percentage of elections in which the winner sets contained two different winners, assuming neither method produced a tie. We see, for example, that under either model when $n\in \{4,5,6\}$ and $S=2$, the simulations returned a handful of elections in which the winner sets were disjoint. When $S=3$ the simulations never returned disjoint winner sets, showing that elections like what  are constructed in Proposition \ref{first_prop} are very rare.

 \begin{table}[htb]
  \centering
  
   \begin{tabular}{c|>{\centering\arraybackslash} m{1.8 cm}|>{\centering\arraybackslash} m{1.8 cm}|>{\centering\arraybackslash} m{1.8 cm}|>{\centering\arraybackslash} m{1.8 cm}|>{\centering\arraybackslash} m{1.8 cm}}
\multicolumn{6}{c}{Impartial Culture }\\
  \hline
  \hline
  
 $n,S$ & Same Winners & 1 Winner Diff. & 2 Winners Diff. & STV chooses CC & RCV chooses CC\\
  \hline
  
 $3,2$ &78.8   &21.2  & 0 &76.6  & 99.9\\
  \hline
  $4,2$ &68.1& 31.8& 0.08& 65.6&94.0  \\
  $4,3$ &65.3 &35.7  & 0 &65.2  &99.7 \\
  \hline
   $5,2$ &61.3  &38.6  &0.12  &58.2  &88.2 \\
   $5,3$ & 49.3& 48.4 & 2.1 &49.7  &92.7 \\
   $5,4$  & 56.1 & 43.8 & 0 &57.2  &99.6 \\
   \hline
  $6,2$ &56.6 & 43.2 &0.006 &53.4 &83.7 \\
   $6,3$ &  40.6& 54.5 & 4.5 & 41.3 &86.0 \\
   $6,4$&  37.9& 56.5 &5.3  &39.4  &92.0 \\
   $6,5$& 49.1 & 50.9 & 0 &50.4  &99.5 \\
   \hline
       \end{tabular}
  
  \vspace{.2 in}
  \begin{tabular}{c|>{\centering\arraybackslash} m{1.8 cm}|>{\centering\arraybackslash} m{1.8 cm}|>{\centering\arraybackslash} m{1.8 cm}|>{\centering\arraybackslash} m{1.8 cm}|>{\centering\arraybackslash} m{1.8 cm}}

\multicolumn{6}{c}{Impartial Anonymous Culture }\\
  \hline
  \hline
  
 $n,S$ & Same Winners & 1 Winner Diff. & 2 Winners Diff. & STV chooses CC & RCV chooses CC\\
  \hline
  
 $3,2$ & 81.3 & 18.3 & 0 &80.8  & 99.9\\
  \hline
  $4,2$ &  65.1& 34.1 & 0.07 &63.3 &94.4 \\
  $4,3$ & 66.9& 32.7 & 0 &67.8  & 100.0\\
  \hline
   $5,2$ &  58.6& 39.3 & 0.04 &56.2  & 88.8\\
   $5,3$ &  47.2& 49.1 &2.3  & 48.1 & 93.3\\
   $5,4$  &  55.8&42.8  &  0&57.3  &99.9 \\
   \hline
   $6,2$ & 54.5& 41.5 & 0.003 &51.7 &83.8 \\
   $6,3$ & 39.0& 53.6 & 4.6 &40.0  &86.3 \\
   $6,4$& 35.8 & 55.1 &5.6  &37.6  &92.2 \\
   $6,5$& 46.6& 50.0& 0 & 48.2 &99.7 \\
   \hline
    \end{tabular}

  \caption{Simulation results using 100,000 runs under the IC (top) and IAC (bottom) models with 1,001 voters. Each entry in the table is the percentage of runs with the corresponding outcome. The Condorcet committee (CC) percentages are conditioned on the existence of a Condorcet committee.}
  \label{IC_IAC_simulations}
  \end{table}

The general takeaway from our simulations is that the two voting methods frequently produce different winner sets under random models, producing the same winners with a probability significantly less than 50\% for some choices of $n$ and $S$. For the smaller choices of $n$ that we use, when the winner sets disagree they tend to differ by only one candidate, but the table suggests that as $n$ increases the sets will differ by more than one candidate. Furthermore, sequential RCV tends to select the Condorcet committee significantly more often than STV.  We note that under random models Condorcet committees often do not exist; the most extreme case is $n=6$, $S=3$, in which such a committee exists in only approximately 42\% of the elections generated under the IC model. This is in stark contrast to real-world elections, in which a Condorcet committee almost always exists (see Sections \ref{scottish_data} and \ref{utah_data}).

It is interesting that the two models of voter behavior generally give similar percentages for each outcome in Table \ref{IC_IAC_simulations}. In other settings the two models can make very different predictions. For example, if the electorate is very large then under the IAC model the probability that a three-candidate election produces a candidate with a majority of first-place votes is approximately 9/16 (this can be calculated using a software package like Normaliz); the corresponding probability under the IC model is approximately zero (this can easily be verified using a binomial distribution).

Under both models, the degrees of misrepresentation and maximal representation tend to be higher under sequential RCV than STV, as expected. For example, Figure \ref{IAC_degrees} shows the distribution of the two degrees for the two methods in the 34166 elections in which the two methods produced different winner sets under the IAC model for the case $n=4$ and $S=2$. The degree of maximal representation in particular tends to be much larger under sequential RCV than STV. Table \ref{degrees_simulations} gives the average degree of each type for elections in which the two methods produced different winner sets (and neither method produces a tie) in the simulations. For example, in the case $n=4$ and $S=3$, under the IAC model the average degree of maximal representation is 33.9\% across the generated elections in which sequential RCV and STV produce different winners sets. For each choice of $n$ and $S$, under either model the average degree of either type is larger under sequential RCV than STV, often much larger.

To conclude this section, we note that we also ran simulations under other choices of $V$ such as 501 and 801. The percentages given by these other choices of $V$ did not markedly differ from the percentages given by $V=1001$, differing only by a percentage point or two. Thus, these results seem to be robust with respect to the number of voters, as long as the electorate is sufficiently large.

\begin{figure}[] 
\begin{center}
\includegraphics[scale=0.45]{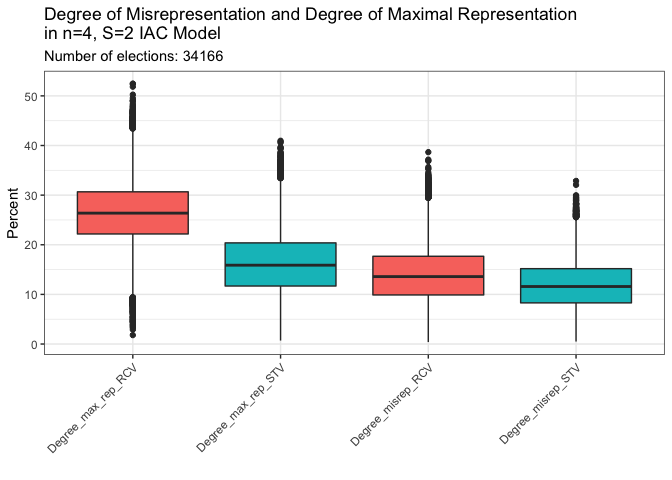}\\

\end{center}
\caption{The degree of misrepresentation and degree of maximal representation under sequential RCV and STV under the IAC model for $n=4$ and $S=2$. The plot was generated using only elections in which the two methods produced different winner sets.}
\label{IAC_degrees}
\end{figure}

 \begin{table}[htb]
  \centering
  
   \begin{tabular}{c|>{\centering\arraybackslash} m{2.5 cm} | >{\centering\arraybackslash} m{2.5 cm}| >{\centering\arraybackslash} m{2.5 cm}| >{\centering\arraybackslash} m{2.5 cm}}
\multicolumn{5}{c}{Impartial Culture }\\
  \hline
  \hline
  
 $n,S$ & Average Deg. Mis. RCV & Average Deg. Mis. STV & Average Deg. Max. Rep. RCV & Average Deg. Max. Rep. STV\\
  \hline
  
 $3,2$ & 0 & 0 & 34.9 & 32.9 \\
  \hline
  $4,2$&   16.2 & 15.8& 17.9 & 16.7\\
  $4,3$& 0&0& 26.3 & 24.7 \\
  \hline
   $5,2$ & 29.0 & 28.5 & 11.0 & 10.1 \\
   $5,3$ &  0&0&10.8&10.1\\
   $5,4$  & 0&0&21.1&19.8 \\
   \hline
  $6,2$ & 38.7&38.2 & 7.4 & 6.8 \\
   $6,3$ & 4.7&4.5&5.6&5.1  \\
   $6,4$& 0&0& 7.3&6.7\\
   $6,5$& 0&0&17.7&16.6 \\
   \hline
       \end{tabular}
  
  \vspace{.2 in}
   \begin{tabular}{c|>{\centering\arraybackslash} m{2.5 cm} | >{\centering\arraybackslash} m{2.5 cm}| >{\centering\arraybackslash} m{2.5 cm}| >{\centering\arraybackslash} m{2.5 cm}}
\multicolumn{5}{c}{Impartial Anonymous Culture }\\
  \hline
  \hline
  
 $n,S$ & Average Deg. Mis. RCV & Average Deg. Mis. STV & Average Deg. Max. Rep. RCV & Average Deg. Max. Rep. STV\\
  \hline
  
$3,2$ & 0 & 0 & 48.6 & 28.0 \\
  \hline
  $4,2$&14.0&11.9&26.4&16.3    \\
  $4,3$& 0&0&33.9&23.3 \\
  \hline
   $5,2$ & 27.1&25.7&13.4&10.4 \\
   $5,3$ &  0&0&13.0&10.2\\
   $5,4$  & 0&0&23.6&19.5 \\
   \hline
  $6,2$ & 38.0&37.2 & 8.0&6.9 \\
   $6,3$ & 4.5&4.3&5.9&5.2  \\
   $6,4$& 0&0&7.7&6.8 \\
   $6,5$& 0&0&18.2&16.5 \\
   \hline
       \end{tabular}

  \caption{The average degrees of the two types for the simulated elections in which sequential RCV and STV produced different winner sets.}
  \label{degrees_simulations}
  \end{table}

\section{Empirical Results: Scottish Data}\label{scottish_data}

In this section we present empirical results using a large database of Scottish local government elections. We begin with a description of these elections, including where the data was found. We then provide results in which we use the number of seats which were available in the original election (most elections contained three or four seats), and we conclude by analyzing the case where we use $S=2$ for all elections.

\subsection{Description of Scottish STV data} 
For the purposes of local government, Scotland is partitioned into 32 council areas, each of which is governed by a council. The councils provide a range of public services that are typically associated with local governments, such as waste management, education, and building and maintaining roads. The council area is divided into wards, each of which elects a set number of councilors to represent the ward on the council. The number of councilors representing each ward is determined primarily by the ward's population, although other factors play a role\footnote{For complete details about how the number of councilors for a ward is determined, see \url{https://boundaries.scot/reviews/fifth-statutory-reviews-electoral-arrangements}.}. In the vast majority of wards, $S=3$ or $S=4$: five elections satisfy $S=2$, 554 satisfy $S=3$, 508 satisfy $S=4$, and three satisfy $S=5$. Every five years each ward holds an election in which all seats available in the ward are filled using the method of STV.

Every Scottish ward has used STV for local government elections since 2007. Preference profiles from the 2007 elections are mostly unavailable. We obtained 2007 ballot data for the 21 ward elections of the Glasgow City Council area from preflib.org \cite{MW}, but we could not find any other preference profiles from the 2007 elections.  We contacted several council election offices and either received no response or were told that the 2007 data is not available.  We obtained preference profile data for the 2012 and 2017 ward elections from the Local Elections Archive Project \cite{T}, although some of this data is still available on various council websites. Some of these profiles required cleaning, and for the 2012 profiles we had to add each candidate's party identification to the data files. The first author obtained data for the 2022 preference profiles directly from the  council websites or by request. We also obtained one multiwinner by-election which occurred in an off-cycle election year. The full dataset is available by request.

In all, we collected the preference profile data of 1,070 multiwinner STV elections from Scotland.  The size of this database generally dwarfs the sizes of databases used in other empirical social choice studies which use preference data from actual elections. The reason is that in most jurisdictions which use STV (such as Australia, Ireland, Northern Ireland, New Zealand, etc.) the ballot data is not made available. We note that there are a few multiwinner preference profiles available from STV political elections in the United States \cite{O}, but we prefer not to include these elections in our Scottish database because as part of our results we analyze the effect of the voting method on the party composition of the winner set, and the American elections for which data is available are non-partisan.

\subsection{Scottish data results}\label{scottish_subsection} We provide a concise summary of the results, providing two examples for context; an interested reader can see complete details for all elections at \cite{Git}.

Of the 1,070 elections, the two voting methods choose different winner sets in 485 of them, so that the methods disagree in 45.3\% of the elections. In 474 of these 485 elections, the winner sets differed by a single candidate while in the remaining 11 elections the winner sets differed by two. None of the elections produced disjoint winner sets. There were 1052 elections containing a Condorcet committee of size $S$ and sequential RCV selects this committee in 1024 of these elections. Thus, in practice sequential RCV  chooses the Condorcet committee with a probability of 97.3\%, assuming such a committee exists. This probability falls to 56.1\% for STV. In the real-world elections in our database, sequential RCV almost always selects the Condorcet committee while STV does not, which accounts for most of the disagreement between the two methods.

A visualization of the degrees of misrepresentation and maximal representation for the 485 elections in which the methods produce different winner sets is given in Figure \ref{scottish_plot1}. Note that both degrees are generally much higher for sequential RCV than STV, which is to be expected if sequential RCV is an excellence-based method.

\begin{figure}[] 
\begin{center}
\includegraphics[scale=0.45]{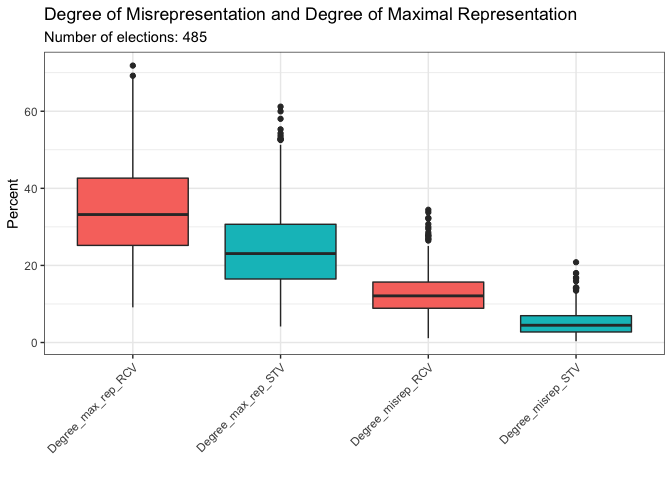}\\

\end{center}
\caption{The degree of misrepresentation and degree of maximal representation under sequential RCV and STV across the 485 elections in which the methods produce different winner sets.}
\label{scottish_plot1}
\end{figure}

Interestingly, there are 19 elections in which the degree of misrepresentation under STV is higher than under sequential RCV, and thus STV does not always perform better under this measure. In these elections the difference between the degrees from the two methods are generally quite small, usually at most one percentage point. The largest difference occurs in the 2012 Ward 12 election of the South Lanarkshire Council area, in which the degree of misrepresentation under sequential RCV (respectively STV) is 11.0\% (respectively 13.3\%).

The Scottish elections are \emph{partisan}, meaning that each candidate officially runs as a member of a party or as an independent. The choice of voting method can significantly affect the party composition of the winner set and thus we summarize the party dynamics under both methods. In 293 elections, STV results in more parties being represented in the winner set than sequential RCV. There are two elections in which  STV results in two more parties being represented in the winner set than what would occur under sequential RCV. There are no elections in which  sequential RCV results in a larger number of parties in the winner set than STV. As with the Condorcet dynamics, we see that sequential RCV does not attempt to achieve proportional representation, instead choosing candidates that are more similar to each other than candidates who would be chosen by STV.

To conclude this section we provide two examples of Scottish elections which provide context for the above results. We use the following abbreviations for Scottish political parties: Conservative (Con), Green (Grn), Independent (Ind), Labour (Lab), Liberal Democrats (LD), and Scottish National Party (SNP).

\begin{ex}\label{Fife_ex}
In the 2012 election of the St. Andrews ward of the Fife Council area, the actual winner set under STV is $\{$McCartney (SNP), Melville (LD), Morrison (Con), Thomson (Lab)$\}$. When using STV the ward is represented by four different parties. Under sequential RCV the winner set is $\{$Melville (LD), Morrison (Con), Sangster (LD), Waterston (LD)$\}$, which is the Condorcet committee. In this case only two parties are represented in the winner set, as the SNP and Labour candidates are replaced by Liberal Democrats. As shown in Table \ref{votes_Fife}, Sangster and Waterston do not have strong electoral support from the perspective of STV, but when using sequential RCV the Liberal Democrats win the first three seats while the fourth seat goes to Morrison. Table \ref{votes_Fife_RCV} shows the votes-by-round tables for the allocation of the first two seats under sequential RCV; note that Sangster essentially just takes the place of the first-seat winner Melville when allocating the second seat.  The same dynamic plays out for Waterston when allocating the third seat. It seems that the Liberal Democrats occupy a middle ground between the Conservative Party and the SNP, so that in the penultimate round the candidate from one of those parties is eliminated and many more votes flow to the remaining Liberal Democrat than to the other surviving candidate. Perhaps if there were a fourth Liberal Democrat in the election, all four seats would have gone to that party under sequential RCV.

In this election, the degree of misrepresentation under sequential RCV is 24.8\% while under STV it is only 8.2\%. Unsurprisingly, when we select three Liberal Democrats instead of candidates from four different parties, the percentage of voters who are not represented by any of the candidates ranked in their top four is much lower. By contrast, the degree of maximal representation under sequential RCV is 21.4\%, higher than the corresponding degree of 15.5\% under STV.
\end{ex}

\begin{table}[htb]
  \centering
  
   \begin{tabular}{c|cccccccc}
  \hline
  Candidate & \multicolumn{8}{c}{Votes by Round}\\
  \hline
  Bridgman (Grn) & 437 & 453 & 464 & & & & & \\
  MacDonald (Ind) & 369 & 453 & 475 & 534 & 569 & 624.3 & 626.1 &    \\
  McCartney (SNP) & 690 & 706 & 725 & 806 & 838 & 867.1 & 867.4 & \textbf{979.3}\\
  Melville (LD) & 641 & 665 & 772 & 858 & \textbf{1189} &&\\
  Morrison (Con) & 769 & 796 & 820 & 843 & \textbf{889} &&\\
  Paul (Ind) &256 &&&&&&&\\
  Sangster (LD) &298 & 311&&&&&&\\
  Thomson (Lab) &600 & 629 & 647 & 742 & 776 & 828.8 & 829.2 & \textbf{951.2}\\
  Waterston (LD) &351 & 381 & 466 & 518 &&&\\
  \hline
  \end{tabular}
  
  \caption{The votes-by-round table for 2012 election of the St. Andrews ward of the Fife Council area, in which quota is 883.}
  \label{votes_Fife}
  \end{table}

\begin{table}[htb]
  \centering
  
   \begin{tabular}{c|cccccccc}
  
  \multicolumn{9}{c}{First Seat, Sequential RCV}\\
  \hline
  \hline
  Candidate & \multicolumn{8}{c}{Votes by Round}\\
  \hline
  Bridgman (Grn) & 437 & 453 & 464 & & & & & \\
  MacDonald (Ind) & 369 & 453 & 475 & 534 &  569 &   &  &    \\
  McCartney (SNP) & 690 & 706 & 725 & 806 & 838 & 925 & 1063 &1123 \\
  Melville (LD) & 641 &       665 & 772 & 858 &1189 & 1317&1501& \textbf{1882}\\
  Morrison (Con) & 769 &   796 &  820 & 843 & 889 & 990 &1062&\\
  Paul (Ind) &256 &&&&&&&\\
  Sangster (LD) &298 &     311&&&&&&\\
  Thomson (Lab) &600 &    629 &   647 & 742 & 776 & 843 &  & \\
  Waterston (LD) &351 &   381 &   466 & 518 &&&\\
  \hline
  \end{tabular}
  
  \vspace{.1 in}
  
     \begin{tabular}{c|cccccccc}
  
  \multicolumn{8}{c}{Second Seat, Sequential RCV}\\
  \hline
  \hline
  Candidate & \multicolumn{7}{c}{Votes by Round}\\
  \hline
  Bridgman (Grn) & 460 & 482 & 527 & & & &  \\
  MacDonald (Ind) & 392 & 479 &  &  &   &     &    \\
  McCartney (SNP) & 711 & 731 & 793&  887&  930&  1070&   \\
  $\xcancel{\text{Melville (LD)}}$ &  &          &  &  &  &  &  \\
  Morrison (Con) & 852 &   796 &  974& 1016&  1092&1166 &1271\\
  Paul (Ind) &278 &&&&&&\\
  Sangster (LD) &608 &     628& 722&811&1191& 1392&\textbf{2120}\\
  Thomson (Lab) &620 &    654 &  702 & 820&  868&  &   \\
  Waterston (LD) &467 &   506 &   553 & 629 &&\\
  \hline
  \end{tabular}

  \caption{(Top) The votes-by-round table for the first seat under sequential RCV in the 2012 election of the St. Andrews ward of the Fife Council area. (Bottom) The votes-by-round table for the second seat under sequential RCV.}
  \label{votes_Fife_RCV}
  \end{table}

The next election gives the largest difference in our database of the degree of maximal representation between the two methods.

\begin{ex}\label{perth-kinross_ex}
In the 2022 election of the Strathtay ward of Perth and Kinross Council area, the actual winner set under STV is $\{$James (Con), Laing (SNP), McLaren (LD)$\}$ while the winner set under sequential RCV is $\{$Kinney (SNP), Laing (SNP), McLaren (LD)$\}$, which is the Condorcet committee. The degree of maximal representation under sequential RCV is 54.5\%, which drops to only 10.4\% under STV. This kind of outcome is what opponents of sequential RCV fear \cite{H}: that a majority of the electorate can choose the entire winner set at the expense of minority interests. The election contained only five candidates-two Conservatives, one Liberal Democrat, and two members of the SNP. In Scotland, Conservatives are generally seen as being on the right side of the political spectrum while the Liberal Democrats and the SNP are generally seen as being to the left, and thus the candidate pool displays a sharp political divide. The two Conservatives earn 34.9\% of the first-place votes and therefore, under proportional representation, one of them should probably earn one of the three seats. However, the remaining 65.1\% of the electorate generally support non-right candidates, and under sequential RCV the more left-wing part of the electorate dictates the composition of the entire winner set. Under STV the more right-wing part of the electorate achieves some representation, as James (Con) achieves quota after receiving a large vote transfer from the elimination of the other Conservative.

The degree of misrepresentation under sequential RCV is 15.3\%, about 12 times higher than the corresponding degree of 1.3\% under STV.
\end{ex}

\subsection{$S=2$ results}

The Utah elections (see Section \ref{utah_data}) which use sequential RCV almost always satisfy $S=2$, and thus we repeat the analysis of Section \ref{scottish_subsection} using two seats in each election. This is perhaps the best way to use the Scottish data to try to predict what might occur in Utah elections if the method continues to be used. The drawback of this analysis is that for most of our original elections $S>2$ and we cannot know if voters would cast the same preference ballots if fewer seats were available. However, our ballot data represents preferences provided in a real election and such a large database of elections should still have some relevance for the $S=2$ setting. Complete details of these results are available at \cite{Git}.

If we set $S=2$ for all 1070 elections in the database, 593 produce different winner sets under the two methods, giving a much higher rate of disagreement than when we use the actual number of seats. When the winner sets are different they differ by only one candidate, as none of the elections produce disjoint winner sets. There are 1045 elections which contain a Condorcet committee of size 2 and in 979 of these sequential RCV selects this committee. This number drops to 437 for STV.  When $S=2$ both methods select the Condorcet committee with a lower frequency than in the original data, but the drop for STV is larger.

The degrees of misrepresentation and maximal representation for the two methods when $S=2$ is given in Figure \ref{scottish_plot2}. For both methods the degree of misrepresentation increases from Figure \ref{scottish_plot1}, presumably because with fewer seats available it is difficult to provide a broad base of representation to the electorate. The degree of maximal representation when $S=2$ provides a stark contrast between the two methods. The degree for sequential RCV does not markedly change from Figure \ref{scottish_plot1}, but the degree for STV drops significantly. The reason for this drop seems to be that STV chooses two candidates who are quite different from each other in an attempt to represent as many voters as possible. When there are more seats available (Figure \ref{scottish_plot1}), STV is able to provide more voters with everything they want, probably in part due to the prevalence of partial ballots in the data.

The party dynamics of the two methods for the $S=2$ case are unsurprising given previous results. The winner set under sequential RCV (respectively STV) selects winners from two different parties in 607 (respectively 1040) elections. STV is much better at providing diversity of party representation, generally producing winner sets with two candidates who are not extremely similar to each other.

\begin{figure}[] 
\begin{center}
\includegraphics[scale=0.45]{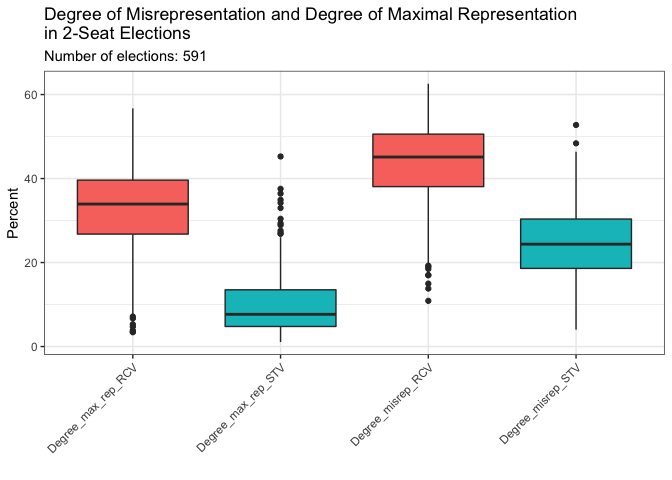}\\

\end{center}
\caption{The degree of misrepresentation and degree of maximal representation under sequential RCV and STV across the 592 elections in which the methods produce different winner sets when $S=2$.}
\label{scottish_plot2}
\end{figure}

\section{Empirical Results: Utah Data}\label{utah_data}

As mentioned in the introduction, since 2019 several cities in Utah have used sequential RCV for multiwinner city council elections. The ballot data is available for eight of these elections \cite{O}. We provide the same analysis for these elections as we did for the Scottish elections, with the exception that we do not analyze party dynamics because the Utah elections are non-partisan. The results are shown in Table \ref{utah_results}. Eight is not a large sample size but the results are consistent with the results of Sections \ref{simulations} and \ref{scottish_data}. Sequential RCV and STV disagree in three of the eight elections. When there is disagreement, sequential RCV chooses the Condorcet committee, which exists in all of these elections and is always chosen by sequential RCV. When the winner sets disagree, the degree of misrepresentation and the degree of maximal representation are smaller for STV than sequential RCV.

 \begin{table}[htb]
  \centering
  {\small
\begin{tabular}{c|c|c|c|c}
 Election &STV Winners & RCV Winners & Mis. Deg. & Max. Rep. Deg.\\
\hline
 2019 Payson & Carter, Hulet, Welton & Carter, Hulet, Welton & 1.4, 1.4 & 44.2, 44.2\\
2019 Vineyard & Flake, Welsh & Flake, Welsh & 18.4, 18.4 & 23.1, 23.1 \\
 2021 Genola  & \textbf{Hughes}, Robison & \textbf{Lundberg}, Robison & 3.4, 1.1 & 46.3, 36.0\\
 2021 Lehi  & Condo, \textbf{Miles} & Condo, \textbf{Hancock}& 47.1, 39.1 & 20.0, 3.3\\
 2021 Moab  &Taylor, Wojchiechowski &Taylor, Wojchiechowski &5.6, 5.6 & 17.5, 17.5\\
2021 Springville  & \textbf{Conover}, Nelson & \textbf{Jensen}, Nelson &35.6, 28.8 & 22.2, 11.0  \\
 2021 Vineyard  &Rasmussen, Sifuentes &Rasmussen, Sifuentes &4.0, 4.0 & 36.3, 36.3\\
 2021 Woodland Hills&Kynaston, Lunt &Kynaston, Lunt &18.3, 18.3&23.1, 23.1  \\

\end{tabular}
}
\caption{The results of the 8 elections in Utah which used sequential RCV. In each election the Condorcet committee of size $S$ equals the winner set under sequential RCV. The degrees of misrepresentation or maximal representation are given in the form RCV, STV.}
  \label{utah_results}
  \end{table}

The election which demonstrates the most extreme difference between the two methods is the 2021 city council election in Lehi, which we analyze in the following example. We cannot provide the preference profile for this election because there were more than 7000 different types of ballots cast (the election contained nine candidates and voters could cast partial ballots, resulting in an extremely large set of possible ballots).

\begin{ex}

\textbf{Lehi, UT 2021 council election.} The actual winner set is $\{$Condo, Hancock$\}$, which is the Condorcet committee and produces a degree of maximal representation of 20\%. If STV were used then the winner set would be $\{$Condo, Miles$\}$ and only 3.3\% of the electorate would get their top two candidates on the council. The degree of maximal representation under sequential RCV is approximately six times larger than under STV. This occurs because Condo is the IRV winner and voters tend to give Condo and Hancock similar rankings on their ballots: in approximately 41.6\% of the ballots Condo and Hancock are ranked consecutively\footnote{We do not count a ranking as ``consecutive'' if one of the candidates is ranked last on a partial ballot and the other is not ranked.}. The next ``most similar'' candidate to Condo is Kunze, who is given a consecutive ranking to Condo on 19.8\% of the ballots. The other STV winner Miles is the fifth-most similar candidate to Condo under this measure, achieving consecutive rankings on only 14.2\% of the ballots.

\end{ex}

The 2021 city council election in Genola contained only three candidates and thus has a small enough preference profile to display (Table \ref{genola_profile}). While the difference between sequential RCV and STV is not as extreme as the Lehi election, the dynamics are similar. Lundberg and Robison are ``similar'' candidates in the way Condo and Hancock are similar (although not to the same degree), and thus Lundberg wins the second seat after Robison (the IRV winner) wins the first under sequential RCV.  

 \begin{table}[htb]
  \centering
  
  \begin{tabular}{l|ccccccccc}
  Num. Voters&13&58& 22&4 & 24 & 60 & 15 & 86 & 96\\
  \hline
  1st choice  & H & H  & H & L & L & L & R & R & R\\
  2nd choice &    & L   & R  &   & H & R &   & H & L\\
  3rd choice  &    & R  & L  &    & R & H &   & L & H\\
  \end{tabular}
  
   \caption{The preference profile for the 2021 city council election in Genola, UT.}
  \label{genola_profile}
  \end{table}


\section{Conclusion}\label{conclusion}

Our results show that sequential RCV is a voting method designed for an excellence-based electoral setting, and therefore generally produces different outcomes than STV. The Monte Carlo simulations and the empirical analysis using elections from Scotland and Utah demonstrate that sequential RCV frequently chooses a different winner set than STV, and this difference is often the result of sequential RCV selecting the Condorcet committee while STV chooses a winner set that aligns with the goal of proportional representation. Our empirical results show that when the two voting methods disagree, generally the degrees of misrepresentation and maximal representation are much higher for sequential RCV than STV. This shows that sequential RCV often creates a winner set with candidates that are more similar to each other (in some sense) than the winner set under STV. 

Consequently, if jurisdictions are considering the adoption of a ranked-choice voting method for multiwinner elections (as is the case for many cities in the United States), we argue that elections officials should not begin by considering specific methods.  Instead, they must first consider the more foundational question ``do we want a method that chooses the best $S$ candidates or a method designed for proportional representation?'' The cities in Utah which use sequential RCV have opted for the former, and we could find no documentation which makes clear that they are aware they made this choice. 

Other jurisdictions are aware of this choice, and thus the analysis in this article is useful for backing up their claims. For example, the city of Northampton, MA, formed a committee in 2021 to explore the possibility of using ranked-choice voting for single-winner and multiwinner municipal elections. In January 2022 the committee circulated a document \cite{Northampton_doc} which compares different choices of multiwinner voting methods, and this documents notes that STV achieves ``Fair representation of voter diversity'' (which we take to be synonymous with proportional representation) while sequential RCV does not. The authors opine that sequential RCV therefore is most likely not an appropriate voting method for city council elections. As a result, the committee ultimately recommended the use of STV \cite{Northampton_rec}.

As a final remark, we note that sequential RCV may be an appropriate method in a political context if the elections are non-partisan and the electorates are not particularly polarized. Utah is generally a politically conservative state and the cities which use sequential RCV are small, suggesting that the electorates may be non-polarized and simply want conservatives to fill the available city council seats. In such a setting, it is perhaps not unreasonable to elect the two or three ``best'' conservatives to the council, filling the seats with candidates who are similar to each other. Of course, even in an excellence-based setting sequential RCV may not be a good method to use, depending on one's criteria for evaluating voting methods. We suspect that many social choice theorists would argue that excellence-based methods like $k$-Borda scoring methods are preferable to sequential RCV, for example. Such questions are beyond the scope of this article and represent fertile ground for future research, especially if other jurisdictions consider the adoption of sequential RCV.

\end{document}